\documentclass{article}
\usepackage[utf8]{inputenc}

\usepackage{graphics}
\usepackage{epsfig} 
\usepackage{amsmath}
\usepackage{amssymb}
\usepackage{dsfont}
\usepackage{enumerate}
\usepackage[tight]{subfigure}
\usepackage{multirow}
\usepackage{bm}
\usepackage{amsthm}
\usepackage{tikz}
\usepackage{array,makecell}
\usepackage{hyperref}

\newtheorem{definition}{Definition}
\newtheorem{assumption}{Assumption}

\newtheorem{proposition}{Proposition}

\newtheorem{lemma}{Lemma}
\newtheorem{remark}{Remark}
\newtheorem{example}{Example}

\newcommand{\prt}[1]{\left(#1\right)}
\newcommand{\brk}[1]{\left[#1\right]}

\newcommand{\abs}[1]{\left|#1\right|}
\newcommand{\norm}[1]{\abs{\abs{#1}}}

\newcommand{\Ep}[1]{\mathbb{E}\left[#1\right]}

\newcommand{\R}{\mathbb R}

\newcommand{\I}{\mathcal{I}}

\title{On Lyapunov functions for open Hegselmann-Krause dynamics}

\author{Renato Vizuete\footnote{ICTEAM institute, UCLouvain, B-1348, Louvain-la-Neuve, Belgium \linebreak (e-mail: renato.vizueteharo@uclouvain.be)}\qquad Paolo Frasca\footnote{Univ.\ Grenoble Alpes, CNRS, Inria, Grenoble INP, GIPSA-lab, F-38000 Grenoble, France (e-mail: paolo.frasca@gipsa-lab.fr)} \qquad Elena Panteley\footnote{Universit\'{e} Paris-Saclay, CNRS, CentraleSup\'{e}lec, Laboratoire des signaux et syst\`{e}mes, 91190, Gif-sur-Yvette, France (e-mail: elena.panteley@l2s.centralesupelec.fr)}}

\date{}

\begin{document}
\maketitle

\begin{abstract}               
In this paper, we provide a formulation of an open Hegselmann-Krause (HK) dynamics where agents can join and leave the system during the interactions. We consider a stochastic framework where the time instants corresponding to arrivals and departures are determined by homogeneous Poisson processes. Then, we provide a survey of Lyapunov functions based on global and local disagreement, whose asymptotic behavior can be used to measure the impact of arrivals and departures. After proving analytical results on these Lyapunov functions in the open system, we illustrate them through numerical simulations in two scenarios characterized by a different number of expected agents.
\end{abstract}

\section{Introduction}

Hegselmann-Krause (HK) model is one of the most important opinion dynamics characterized by bounded confidence interactions \cite{hegselmann2002opinion}. For a population of $n$ agents indexed in a set $\I=\{ 1,\ldots,n\}$, each agent $i$ in the network holds a real-valued opinion $x_i\in\R$ and interacts with other agents only if the difference between their opinions remain inside the confidence interval of the agent, determined by a specific threshold:
\begin{equation}\label{eq:Krause_dynamics}
\dot x_i(t)=\sum_{j:\vert x_i(t)-x_j(t)\vert<1} (x_j(t)-x_i(t)), \quad \text{for all  } i\in\I .
\end{equation}
Since its original formulation, several extensions of this dynamics have been considered to reproduce additional characteristics of more realistic social interactions. For instance, multidimensional dynamics has been formulated in   \cite{nedic2012multi},  heterogeneous threshold were studied in \cite{mirtabatabaei2011opinion} and noisy states have been considered in \cite{pineda2013noisy}. However, changes in the set of agents have not yet been incorporated in the study of the HK dynamics. 

In many social interactions, the group of individuals does not remain constant since new agents may join the system and others leave. This happens specially in interactions taking place over online platforms (Facebook, Twitter, etc.) where agents can connect and disconnect in an easy manner. A similar behavior has been observed at the level of communications networks for the control of connected vehicles \cite{bechihi2023resource}. The importance of these phenomena has been highlighted in \cite{ravazzi2021learning,proskurnikov2018tutorial} as an important feature that must be taken into account in the analysis of dynamic social interactions. 

A multi-agent system characterized by a time-varying set of agents is called \emph{open multi-agent system} (OMAS), where replacements, arrivals and departures can occur~\cite{hendrickx2016open,hendrickx2017open}. Depending on the type of system, different approaches may be used for its analysis, including time-invariant finite superset \cite{vizuete2020influence}, multi-mode multi-dimensional systems \cite{xue2022stability}, or continuum of agents \cite{blondel2010continuous}, etc.

One of the most important challenges in the study of OMAS is the definition of the graph topology during arrivals and departures of agents. In this sense, the HK model avoids any ambiguity since the network topology is automatically defined based on the states of the agents. However, even if the graph topology is well defined, it remains another important question: How to analyze OMAS when the size of the system can change? To the best of our knowledge, there are no results to handle this type of systems. Due to the time-varying dimension, trajectories cannot be well defined due to the lost of the information (departures) and new information coming into the system (arrivals). Several approaches based on the use of scalar functions independent of the dimension of the system have been considered for the analysis of OMAS \cite{monnoyer2020open,vizuete2022contributions}. For relatively simple systems, the choice of this scalar function is evident, as in the case of consensus where the variance is a good measure of disagreement. Nevertheless, the choice is not straightforward for more complex dynamics. For general dynamical systems, Lyapunov functions appear as a potential tool to be used in OMAS. The problem then becomes to measure the changes on properties of the Lyapunov functions due to the arrivals and departures. In the case of the HK dynamics, several Lyapunov functions have been used for the analysis of stability and some of them have completely different behaviors, such that it is possible that only few of them are suitable for the analysis in the open scenario.

Preliminary works in an open HK model have only studied the particular case of replacements where the opinion of agents change abruptly to mimic dynamic groups \cite{torok2013opinions,grauwin2012opinion}. Since the problem is highly complex and asymptotic properties cannot be ensured for the trajectories, most of the previous works rely on simulations of order parameters associated with the dynamics to evaluate the performance under replacements. Regarding stability, the authors in \cite{blondel2010continuous} have considered the addition of new agents only as perturbations to analyze the stability of existing clusters (connected components). 

In this paper, we present a formulation of an open HK dynamics where agents can join and leave the system according to homogeneous Poisson processes, such that the size of the system is time-varying. Based on the characteristics of the HK dynamics, we propose Lyapunov functions for the analysis of the open system, that are formulated as alternative versions of well-known Lyapunov functions for the closed system (no arrivals and departures). We perform an analysis of these functions focused on three important properties: asymptotic value, continuity in time, and monotonicity in time. Finally, relying on numerical simulations we discuss the advantages and drawbacks of the Lyapunov functions in closed and open scenarios.

\paragraph{Outline.} In Section~\ref{sect:closed} we define the relevant Lyapunov functions and discuss their properties, while distinguishing between functions that measure global or local disagreement. In Section~\ref{sect:open} we extend the discussion to Open HK dynamics, by simulations. The final Section~\ref{sect:outro} comments on the opportunities for future research.

\section{Lyapunov functions in closed HK dynamics}
\label{sect:closed}
The differential equation \eqref{eq:Krause_dynamics} usually has no differentiable solutions (classical solutions) since the right-hand side of the equation can be discontinuous when the interaction topology changes, which can prevent $x$ from being differentiable. For this reason, we consider Carathéodory solutions of \eqref{eq:Krause_dynamics}, which correspond to solutions of the integral equation:
\begin{equation}\label{eq:Caratheodory}
    x_i(t)=x_i(t_0)+\int_{t_0}^t \sum_{j:\vert x_i(\tau)-x_j(\tau)\vert<1} (x_j(\tau)-x_i(\tau))  \; d\tau,
\end{equation}
for all $i\in\I$. The set of equilibria of \eqref{eq:Caratheodory} is known to be
 \begin{equation}\label{eq:equilibria_set}
 F=\{x\in \mathbb{R}^n: \text{for every } (i,j)\in \I\times\I,
 \text{ either }x_i=x_j\text{ or }\vert x_i-x_j\vert\ge 1\}
 \end{equation}
 and the following convergence result is available.

\begin{lemma}[\cite{blondel2010continuous}]
For almost every initial condition, there exists a solution for \eqref{eq:Caratheodory} that converges to a limit $x^*\in F$.
\end{lemma}

In this work, we consider that all the initial states of the agents $x_i(t_0)$ belong to an interval $[a,b]$: due to the characteristics of the HK dynamics, they will remain inside this interval during the evolution of the dynamics.

Due to the existence of temporary clusters in the HK dynamics, we will analyze two types of Lyapunov functions based on disagreement of the states: we refer to \emph{global disagreement} when the states of the agents are compared with all the other agents in the system independently of the clusters, and \emph{local disagreement} when the state of the agents are compared only with the states of the agents in the current clusters. 

\subsection{Global disagreement functions}

In this subsection, we consider Lyapunov functions that measure the disagreement among all the agents in the system. A natural candidate to measure the global disagreement is the variance, defined as:
\begin{equation}\label{eq:function_U}
U_0(x):=\dfrac{1}{n}\sum_{i=1}^n (x_i-\bar x)^2,
\end{equation}
where $\bar x=\frac{1}{n}\sum_{i=1}^n {x_i}$ is the average value. This function has been used to prove stability of the equilibria set \eqref{eq:equilibria_set} in \cite{blondel2010continuous}.

Similarly to the variance, we can also define the disagreement function between all the agents of the system through the function:
\begin{equation}\label{eq:function_V}
V_0(x):=\dfrac{1}{n^2}\sum_{i,j=1}^n(x_i-x_j)^2.   
\end{equation}
Additionally, we introduce another classical Lyapunov function, which has been used in  \cite{ceragioli2012continuous}:
\begin{equation}\label{eq:function_T}
T_0(x):=\dfrac{1}{2n}\sum_{i=1}^n x_i^2=\frac{1}{2n}\norm{x}^2,    
\end{equation}
where $\norm{\cdot}$ denotes the Euclidean norm.
Even if it might not be apparent that $T_0(x)$ measures disagreement, it is formulated in terms of global information of the system and, as the following simple result shows, its behavior is equivalent to the functions \eqref{eq:function_U} and \eqref{eq:function_V}.

\begin{proposition}[Relations]\label{prop:equivalence_functions}
For the Lyapunov functions \eqref{eq:function_U}, \eqref{eq:function_V} and \eqref{eq:function_T}, we have:
$$
V_0(x)=2U_0(x)=4T_0(x)-2\bar x^2.
$$
\end{proposition}
\begin{proof}
 First, we find an equivalence between the functions $U_0(x)$ and $T_0(x)$:
\begin{align*}
    U_0(x)
    &=\dfrac{1}{n}\sum_{i=1}^n(x_i^2-2x_i\bar x+\bar x^2)\\
    &=\dfrac{1}{n}\left(\sum_{i=1}^nx_i^2-2\bar x\sum_{i=1}^n x_i+\sum_{i=1}^n\bar x^2\right)\\
    &=\dfrac{1}{n}\left(\sum_{i=1}^nx_i^2-2(n\bar x)\bar x+n\bar x^2\right)\\
    &=2T_0(x)-\bar x^2.
\end{align*}
Next, we find an equivalence between the functions $V_0(x)$ and $U_0(x)$:
\begin{align*}
    V_0(x)
    &=\dfrac{1}{n^2}\sum_{i,j=1}^n(x_i^2-2x_ix_j+x_j^2)\\
    &=\dfrac{1}{n^2}\left(n\sum_{i=1}^nx_i^2-2\left(\sum_{i=1}^n x_i\right)\left(\sum_{j=1}^n x_j\right)+n\sum_{j=1}^nx_j^2\right)\\
    &=\dfrac{1}{n^2}\left(2n\sum_{i=1}^n x_i^2-2n^2\bar x^2\right)\\
    &=2U_0(x), 
\end{align*}
which yields the desired result. 
\end{proof}
The non-increase in time of $U_0(x(t))$ and $T_0(x(t))$ have been proved in \cite{blondel2010continuous} and \cite{ceragioli2012continuous} respectively.

In addition to the previous functions, we also introduce the following Lyapunov function that has been proposed in \cite{piccoli2021generalized} for a generalized model of the HK dynamics:
$$
W_0(x):=\frac{1}{n^2}\prt{\sum_{i,j:\abs{x_i-x_j}<1}\prt{x_i-x_j}^2+\sum_{i,j:\abs{x_i-x_j}\ge 1}1}.
$$
This function is composed by two sums, where the first sum measures the disagreement among the agents that interact, and the second sum counts the agents that do not interact. The first term is thus a measure of ``local'' disagreement: as such, it will be studied in the next section (where it will be denoted as $V$). The second term, however, renders it dependent on  global information and thus a measure of global disagreement. Function $W_0(x(t))$ is continuous but, unlike the previous functions in this section, $W_0(x(t))$ is not differentiable along the trajectories of \eqref{eq:Krause_dynamics}. Further properties of $W_0$ can be deduce from the properties of $V$ studied below.

\subsection{Local disagreement functions}

One of the main drawbacks on the use of Lyapunov functions based on global disagreement is that even if the functions are non-increasing, their asymptotic values are not zero when the system presents several clusters\footnote{In this paper, we call clusters the connected components of the graph associated to the HK dynamics.}. We can see this behavior by analyzing the function $U_0(x(t))$ whose asymptotic value is given by:
$$
\lim_{t\to\infty} U_0(x(t))=\dfrac{1}{n}\sum_{i=1}^K n_i\left(k_i-\dfrac{1}{n}\sum_{j=1}^K n_jk_j \right)^2,
$$
where $K$ is the final number of clusters, $k_i$ is the final value of the agents in cluster $i$ and $n_i$ is the number of agents in the cluster $i$. In fact, we can observe that $\lim_{t\to\infty}U_0(x(t))=0$ only when $K=1$, which means consensus of all the agents of the system, a not so common situation in the HK dynamics. Clearly, from Proposition~\ref{prop:equivalence_functions}, we conclude that the other two functions $V_0$ and $T_0$, do not converge to zero. Similarly, if at the end, we have agents in different clusters, the sum $\sum_{i,j:\abs{x_i-x_j}\ge 1}1\neq 0$ in the function $W_0(x(t))$.

In the context of OMAS, we would like to have functions converging to zero at least for a closed dynamics, such that any deviation from the zero value is due only to the impact of arrivals and departures. For this reason, we bring forward potential Lyapunov functions based on local disagreement.

\subsubsection{Function $V(x)$}
First, let us consider a variant of the function $V_0(x)$ by restricting the disagreement function $V_0(x)$ to the set $D=\{i,j:\vert x_i-x_j\vert<1\}$. Then, we define the function:
\begin{equation}\label{eq:function_V_open}
V(x):=\frac{1}{n^2}\sum_{i,j:\vert x_i-x_j\vert<1}(x_i-x_j)^2.    
\end{equation}
Clearly $V(x)\ge 0$ and it can be shown \cite[Proposition~1]{blondel2010continuous} that along the trajectories $x(t)$ of \eqref{eq:Krause_dynamics}
$$
V(x(t))=-\frac{1}{n}\dfrac{d}{dt}U_0(x(t)).
$$
Since from Theorem 2 in \cite{blondel2010continuous}, every solution $x(t)$ converges to a limit $x^*\in F$, then the function $V(x(t))$ converges to 0.

The function $V(x(t))$ can also be intuitively expressed as 
$$V(x(t))=\frac{1}{n}x(t)^TL_{t}x(t),
$$
where $L_{t}$ is the Laplacian of the interaction graph at time $t$, and its $k$-th derivative (on its domain of definition) can be obtained to be $$\frac{d^{(k)}}{dt^{(k)}}V(x(t))=\frac{(-2)^{k}}{n}x(t)^TL_{ t}^{k+1}x(t).$$

We now consider the monotonicity properties of $V(x(t))$.
\begin{proposition}
With the exception of a countable set of times, $V(x(t))$ is non-increasing.
\end{proposition}
\begin{proof}
Between the times at which a topology change occurs, we have a constant number of clusters $m\le n$. For each cluster we can define a function:
$$
V_{0_\ell}(x)={\sum_{j,k=1}^{c_\ell} (x_j-x_k)^2,}\;\; \ell\in\{1,\ldots,m\},
$$
where $c_\ell$ is the number of agents in cluster $\ell$. Then, the function $V(x(t))$ can be expressed as:
$$
V(x(t))=\frac{1}{n^2}\sum_{\ell=1}^m V_{0_\ell}(x(t))
$$
Since all functions $V_{0_\ell}(x(t))$ are non-increasing (being, up to constants, disagreement functions \eqref{eq:function_V} restricted to the cluster), the proof is completed. 
\end{proof}

Unlike function $V_0(x(t))$, function $V(x(t))$ can be discontinuous and also not monotonic in time, because of the changes in the topology. Indeed, this function can increase during the evolution of the HK dynamics due to the switching topology, as the following example shows. 
\begin{example}
Let us consider three agents $x_1$, $x_2$, $x_3$ such that $x_1<x_2<x_3$, $\vert x_1-x_2\vert<1$, $\vert x_3-x_2\vert<1$ and $\vert x_1-x_3\vert >1$. In this case, agent $x_1$ is interacting only with agent $x_2$, while agent $x_2$ is interacting also with agent $x_3$. Next, we analyze the event characterized by $\vert x_3-x_1\vert=1$ and we denote by $V(x(t^-))$ the value of the function $V(x(t))$ before this event. In this case, the function $V(x(t^-))$ is given by:
$$
V(x(t^-))=(x_1-x_2)^2+(x_2-x_1)^2+(x_2-x_3)^2+(x_3-x_2)^2.
$$
However, after the event, when agent $x_3$ is inside the range of agent $x_1$ and all the three agents are interacting, the function $V(x(t^+))$ is:
\begin{align*}
    V(x(t^+))&=(x_1-x_2)^2+(x_1-x_3)^2+(x_2-x_1)^2\\
    &\quad +(x_2-x_3)^2+(x_3-x_1)^2+(x_3-x_2)^2\\
    &=V(x(t^-))+(x_1-x_3)^2+(x_3-x_1)^2\\
    &=V(x(t^-))+2.
\end{align*}
\end{example}
More generally, we have that when two agents begin to interact, a new edge is added to the interaction graph, and when two agents stop interacting, an edge is removed from the interaction graph. Therefore, the changes of $V(x(t))$ during switching events can be expressed as:
$$
V(x(t^+))=V(x(t^-))+2(e_a-e_r),
$$
where $e_a$ and $e_r$ are the number of edges added and removed respectively during the switching event.

\subsubsection{Function $U(x)$}
Following the approach in \cite{tumash2019synchronization}, we consider the spectral decomposition of the Laplacian matrix:
$$
L=V\Lambda V^T.
$$
The matrix $V$ can be decomposed in two parts as $V=[V_0,V_-]$ where the $n\times m$ and $n\times(n-m)$ matrices $V_0$ and $V_-$ respectively, are constructed from eigenvectors corresponding to zero and negative eigenvalues respectively. Let us define the error as:
$$
\mathrm{e}(x):=V_-V_-^Tx,
$$
and the function:
$$
U(x):=\dfrac{1}{n}\mathrm{e}(x)^T\mathrm{e}(x).
$$

\begin{proposition}

The function $U(x)$ is such that
\begin{equation}\label{eq:alternative_U}
U(x)=\dfrac{1}{n}\sum_{i=1}^n (x_i-\bar x_i)^2,   
\end{equation}
where $\bar x_i$ is the average of the states of the agents in the cluster to which $x_i$ belongs. 
\end{proposition}

\begin{proof}
By definition:
$$
U(x)=\frac{1}{n}\prt{V_-V_-^Tx}^T\prt{{V_-V_-^Tx}}
=\frac{1}{n}\prt{x^TV_-V_-^TV_-V_-^Tx}.
$$
We use the properties $V_-V_-^T=I_n-V_0V_0^T$ and $V_-^TV_-=I_{n-m}$ to obtain:
$$
U(x)=\frac{1}{n}x^TV_-I_{n-m}V_-^Tx
    =\frac{1}{n}\prt{\sum_{i=1}^nx_i^2-x^T\Gamma x},
$$
where $\Gamma$ is a block matrix given by:
$$
\Gamma=\begin{bmatrix}
\Gamma_1 & 0 & \cdots & 0\\
0 & \Gamma_2 & \cdots & 0\\
\vdots & \vdots & \ddots & \vdots\\
0 & 0 & \cdots & \Gamma_m
\end{bmatrix},
$$
and each block is a matrix of ones multiplied by $\frac{1}{m_i}$ where $m_i$ is the size of the block (cluster). Then, we have:
\begin{align*}
    U(x)&=\frac{1}{n}\prt{\sum_{i=1}^n x_i^2-\sum_{j=1}^m m_j\bar x_j^2}\\
    &=\frac{1}{n}\prt{\sum_{i=1}^n x_i^2+\sum_{j=1}^m (m_j\bar x_j^2-2m_j\bar x_j^2)}\\
    &=\frac{1}{n}\prt{\sum_{i=1}^n x_i^2+\sum_{i=1}^n \bar x_i^2-2\sum_{j=1}^m \bar x_j(m_j\bar x_j)}\\
    &=\frac{1}{n}\prt{\sum_{i=1}^n x_i^2+\sum_{i=1}^n \bar x_i^2-2 \sum_{i=1}^n\bar x_i x_i}.
\end{align*}
\end{proof}
Based on this result, 
we can observe that the function $U(x)$ is a variant of the classical Lyapunov function $U_0(x)$ in which the average value is computed only among the agents belonging to the clusters of the system.

When $t\to\infty$, the graph will be partitioned into clusters whose nodes satisfy the condition $\vert x_i(t)-x_j(t) \vert<1$. Since Theorem 2 in \cite{blondel2010continuous} implies that the agents in each cluster of the graph will reach consensus, we have that the value of the mean $\bar x_i$ will coincide with the value of all the agents in the cluster such that:
$$
\lim_{t\to\infty} U(x(t))=0.
$$

In the next proposition we provide an additional proof of the convergence of $U(x(t))$ based on its monotonicity in time.

\begin{proposition}\label{prop:asymptotics_U}
The function $U(x(t))$ along the trajectories $x(t)$ of dynamics \eqref{eq:Krause_dynamics} is non-increasing for all $t$ and
$$
\lim_{t\to\infty}U(x(t))=0.
$$
\end{proposition}
Before presenting the proof of Proposition~\ref{prop:asymptotics_U}, we introduce the following lemma.
\begin{lemma}[\cite{cvetkovski2012inequalities}]\label{prop:partial_means}
Let $a_1, a_2,\ldots,a_n$; $b_1,b_2,\ldots,b_n$ be real numbers such that $b_1,b_2,\ldots,b_n>0$. Then
$$
\dfrac{a_1^2}{b_1}+\dfrac{a_2^2}{b_2}+\cdots+\dfrac{a_n^2}{b_n}\ge \dfrac{(a_1+a_2+\cdots+a_n)^2}{b_1+b_2+\cdots+b_n},
$$
with equality if and only if
$$
\dfrac{a_1}{b_1}=\dfrac{a_2}{b_2}=\cdots=\dfrac{a_n}{b_n}.
$$
\end{lemma}
\begin{proof}[Proof of Proposition~\ref{prop:asymptotics_U}]
Let us consider the function $U(x(t))$ between two possible switching times of dynamics \eqref{eq:Krause_dynamics}. In this case, the function corresponds to the variance of each cluster of the graph, which is non-increasing. Now, we examine the behavior of the function during the switching times. If the clusters remain the same during a switching time, then we have:
$$
U(x(t^+))=U(x(t^-)).
$$
For dynamics \eqref{eq:Krause_dynamics}, the order between agent opinions is preserved and if at some time instant the distance between two consecutive agent opinions $x_i$ and $x_{i+1}$ is larger than or equal to 1, it remains so forever \cite{hendrickx2008graphs}. This implies that two different clusters cannot merge.

Next, we analyze the case when a cluster with $m$ agents is partitioned into $k$ clusters during a switching time. Then, before the switching event, we get:
$$
U(x(t^-))=U^{0}+\sum_{i=1}^m\left(x_i-\frac{1}{m}\sum_{j=1}^m x_j\right)^2=U^{0}+\sum_{i=1}^m x_i^2-\frac{1}{m}\left(\sum_{i=1}^m x_i\right)^2,
$$
where $U^{0}$ corresponds to the clusters without modification. After the switching event, we obtain:
\begin{equation*}
    U(x(t^+))=U^{0}+
    \sum_{i=1}^{m} x_i^2-\left(\!\frac{1}{m_1}\!\left(\sum_{i=1}^{m_1} x_i\!\right)^2\!+\!\cdots\!+\!\frac{1}{m-m_{k-1}}\left(\sum_{i=m_{k-1}+1}^{m} x_i\!\right)^2\right).
\end{equation*}
Then, by using Lemma~\ref{prop:partial_means}, it holds:
\begin{equation}\label{eq:ineq_clusters}
\frac{1}{m}\left(\sum_{i=1}^m x_i\right)^{2}\le \frac{1}{m_1}\left(\sum_{i=1}^{m_1} x_i\right)^{2}+\cdots+\frac{1}{m-m_{k-1}}\left(\sum_{i=m_{k-1}+1}^{m} x_i\right)^{2},    
\end{equation}
such that:
\begin{equation}\label{eq:behavior_U}
U(x(t^-))\ge U(x(t^+)).    
\end{equation}
Since clusters are independent, if more than one cluster is partitioned during the same switching time, the inequality \eqref{eq:ineq_clusters} is valid for each cluster and \eqref{eq:behavior_U} is the only possible behavior for $U(x(t))$ during switching times. 
\end{proof}

An alternative proof of the limit of $U(x(t))$ can be given by using its derivative between switching times, which is negative if $x(t)\notin F$ and zero otherwise:
$$
\frac{d}{dt}U(x(t))=-nV(x(t))=-\frac{1}{n}\sum_{i,j:\vert x_i-x_j\vert<1}(x_i-x_j)^2.
$$

\begin{remark}(Edges in a cluster)
Unlike the function $V(x(t))$, the addition and removal of edges between agents in the same cluster do not modify the value of the function $U(x(t))$.
\end{remark}

\subsection{Comparison between functions}
To summarize the arguments above, we report in Table~\ref{tab:table1} the characteristics of the functions introduced in this work. Since $U_0$, $V_0$ and $T_0$ are equivalent as per Proposition~\ref{prop:equivalence_functions}, we only present the characteristics of $U_0$. In the perspective of studying the Open HK dynamics, the ideal Lyapunov functions should have the properties of being continuous, non-increasing, and converging to zero. Inspecting the table highlights that no available function satisfy this combination of properties.

\begin{table}[]
\centering
 \begin{tabular}{| c | c | c | c |}
 \hline
 \textbf{Function} & \textbf{Limit value} & \makecell{\textbf{Continuity} \\ \textbf{in  time}} & \makecell{\textbf{Monotonicity} \\ \textbf{in time}} \\ 
 \hline
 $U_0(x(t))$ &  $\ne 0$ & Continuous & Non-increasing \\ 
 \hline
 $W_0(x(t))$ & $\ne 0$ & Continuous & Non-increasing \\
 \hline
 $V(x(t))$ & =0 & Discontinuous & Can increase \\
 \hline
 $U(x(t))$ & =0 & Discontinuous & Non-increasing\\ 
 \hline
\end{tabular}
\caption{Characteristics of Lyapunov functions for closed HK dynamics.}
\label{tab:table1}
\end{table}

\section{Lyapunov functions in open HK dynamics}\label{sect:open}

In the open system, the dynamics in continuous time of the agents is characterized by \eqref{eq:Caratheodory}, and the additional changes are due to \emph{arrivals} and \emph{departures}, that generate a time-varying set of agents $\I(t)$.
\begin{definition}[Departure]
We say that an agent $j\in\I(t^-)$ leaves the system at time $t$ if:
$$
\I(t^+)=\I(t^-)\setminus \{j\},
$$
where $\I(t^-)$ is the set of agents before the departure and $\I(t^+)$ is the set of agents after the departure of agent $j$. Thus, $\abs{\I(t^+)}=\abs{\I(t^-)}-1$.
\end{definition}

\begin{definition}[Arrival]
We say that an agent $j$ \footnote{The label of agent $j$ is different from the labels of all the agents that have interacted in the system from time $t_0$ until time $t$.} joins the system at time $t$ if:
$$
\I(t^+)=\I(t^-)\cup \{j\},
$$
where $\I(t^-)$ is the set of agents before the arrival and $\I(t^+)$ is the set of agents after the arrival of agent $j$. Thus, $\abs{\I(t^+)}=\abs{\I(t^-)}+1$.
\end{definition}

In this open system, the solution of each agent $x_i$ satisfies 
\begin{equation}\label{eq:open_Caratheodory}
    x_i(t)=x_i(t_{a_i})+\int_{t_{a_i}}^t \sum_{j\in \I(\tau):\vert x_i(\tau)-x_j(\tau)\vert<1} (x_j(\tau)-x_i(\tau))  \; d\tau,
\end{equation}
for all $t\in\brk{t_{a_i},t_{d_i}}$, where $t_{a_i}$ and $t_{d_i}$ are the time instants of the arrival and departure, respectively, of agent $i$. If agent $i$ was already present at time $t_0$, we consider $t_{a_i}=t_0$. If agent $i$ never leaves the system, we consider that the solution \eqref{eq:open_Caratheodory} is valid for all $t\ge t_{a_i}$.

To analyze this open system, we consider a stochastic setting where the time instants corresponding to arrivals and departures are determined by homogeneous Poisson processes. Following an approach similar to \cite{monnoyer2020open}, we make the following assumptions about the occurrence of departures and arrivals.

\begin{assumption}[Departure process]\label{ass:departure_poisson_process}
The departure instant of an agent $j$ is determined by a homogeneous Poisson process $N_t^{(j)}$ with rate $\lambda_{d}>0$ associated with the agent. Thus, all the departure instants in the system are determined by a Poisson process $N_t^{(D)}$ with rate $\lambda_D(t)=\lambda_dn(t)$ \footnote{The global Poisson process $N_t^{(D)}$ associated with all the arrivals is not homogeneous since the rate $\lambda_D$ is time-varying.}.  
\end{assumption}

\begin{assumption}[Arrival process]\label{ass:arrival_poisson_process}
The arrival instants are determined by a homogeneous Poisson process $N_t^{(A)}$ with rate $\lambda_A>0$. 
During an arrival, an agent $j$ joins the system with a state $x_{j}$ determined by a random variable $\Theta$, which takes values according to a continuous distribution with  support in the interval $[a,b]$, with mean $m$ and variance $\sigma^2$.    
\end{assumption}

In OMAS, the size of the systems is time-varying and even if several Lyapunov functions can be used to prove stability of the HK dynamics, some of them may not be adequate to evaluate the impact of arrivals and departures. It is obvious that in OMAS, the Lyapunov functions cannot be continuous since when an agent joins or leaves the network, the energy of the system is abruptly modified, generating a jump (discontinuity). 

\subsection{Simulations of Open HK dynamics}
We perform numerical simulations of an open HK dynamics to evaluate the behavior of the Lyapunov functions subject to arrivals and departures.
We consider a system composed by $n_0=10$ agents whose initial values are drawn from a uniform distribution $\mathcal{U}[0,6]$. For the closed system, the number of agents is constant (i.e., $n(t)=n_0$ for all $t$) and we compute the expected values of the Lyapunov functions $U_0(x(t))$, $U(x(t))$ and $V(x(t))$ considering 10000 realizations of the process. Then, we perform the simulations of an open HK dynamics with $\lambda_A=5$ and two different values of $\lambda_d$. From the theory of birth-death processes, the asymptotic value of the expected number of agents is given by \cite{monnoyer2020open}:
$$
\lim_{t\to\infty} \Ep{n(t)}=\frac{\lambda_A}{\lambda_d}.
$$
In the first scenario, we choose $\lambda_d=0.4$, which gives us $\lim_{t\to\infty} \Ep{n(t)}=12.5$, such that asymptotically, the expected number of agents is greater than the initial number $n_0$. In the second scenario, we use $\lambda_d=0.62$, which corresponds to $\lim_{t\to\infty} \Ep{n(t)}=8.06$, such that we expect a smaller average number of agents in the long run. The expectation of the Lyapunov functions are again computed over 10000 realizations of the process.

\begin{figure}
\centering
\includegraphics[width=0.7\linewidth]{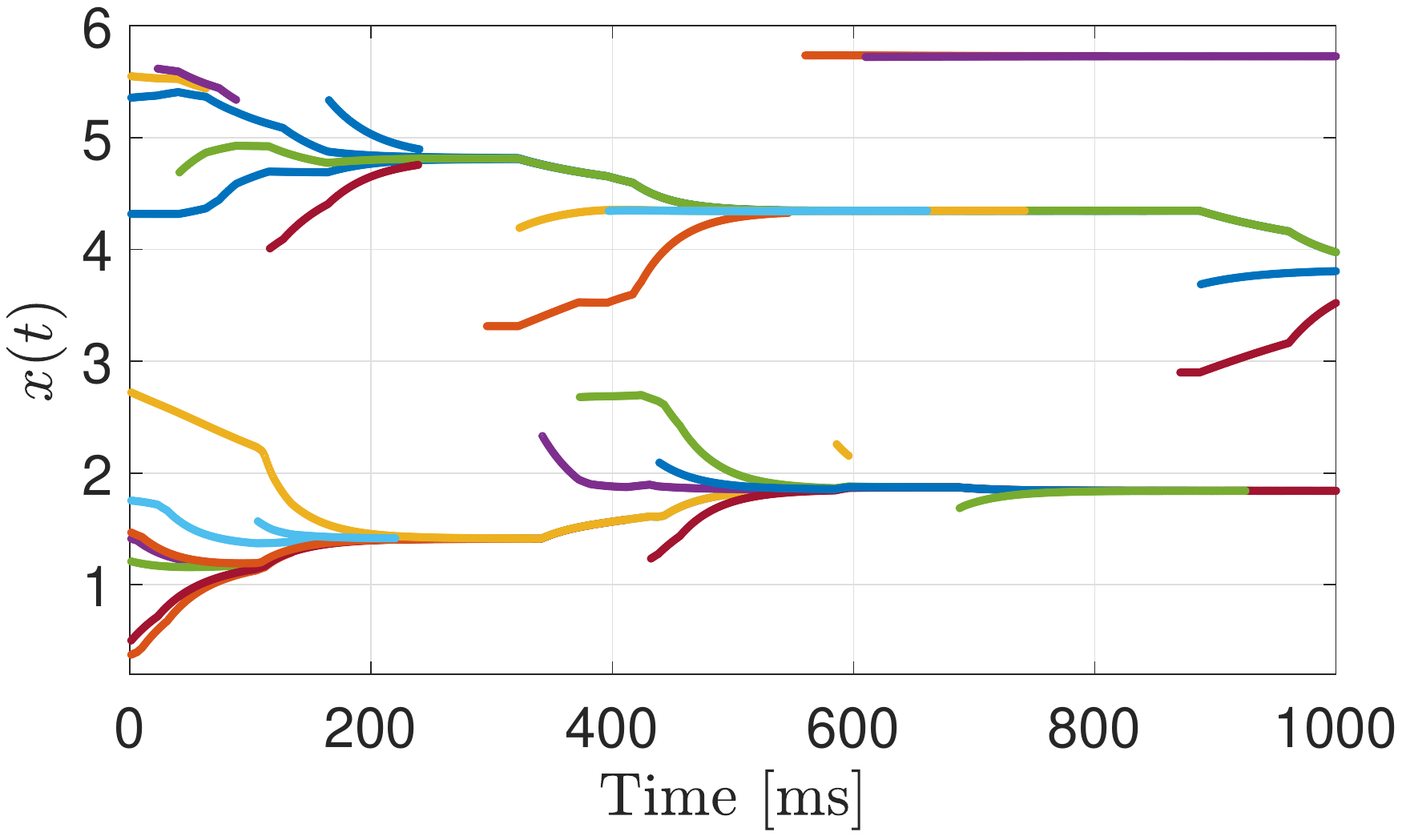}
\caption{Evolution of the opinions for a realization of an open HK dynamics with $n_0=10$ agents and joining and leaving agents following homogeneous Poisson processes with $\lambda_A=5$ and $\lambda_d=0.4$.}\label{fig:x_t_poisson}
\end{figure}

\begin{figure}
\centering
\subfigure[$\Ep{U_0(x(t))}$]
{ \includegraphics[width=0.8\linewidth]{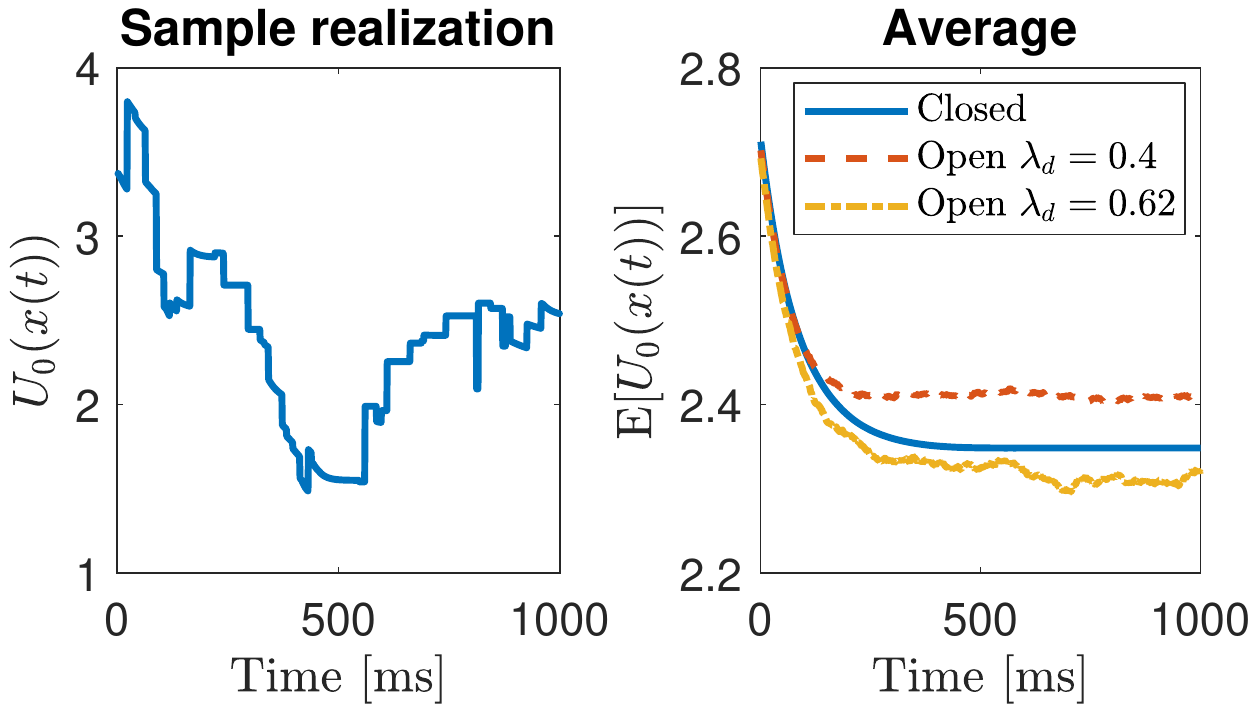}\label{fig_U_0_x_expectation}
}
\subfigure[$\Ep{U(x(t))}$]
{ \includegraphics[width=0.8\linewidth]{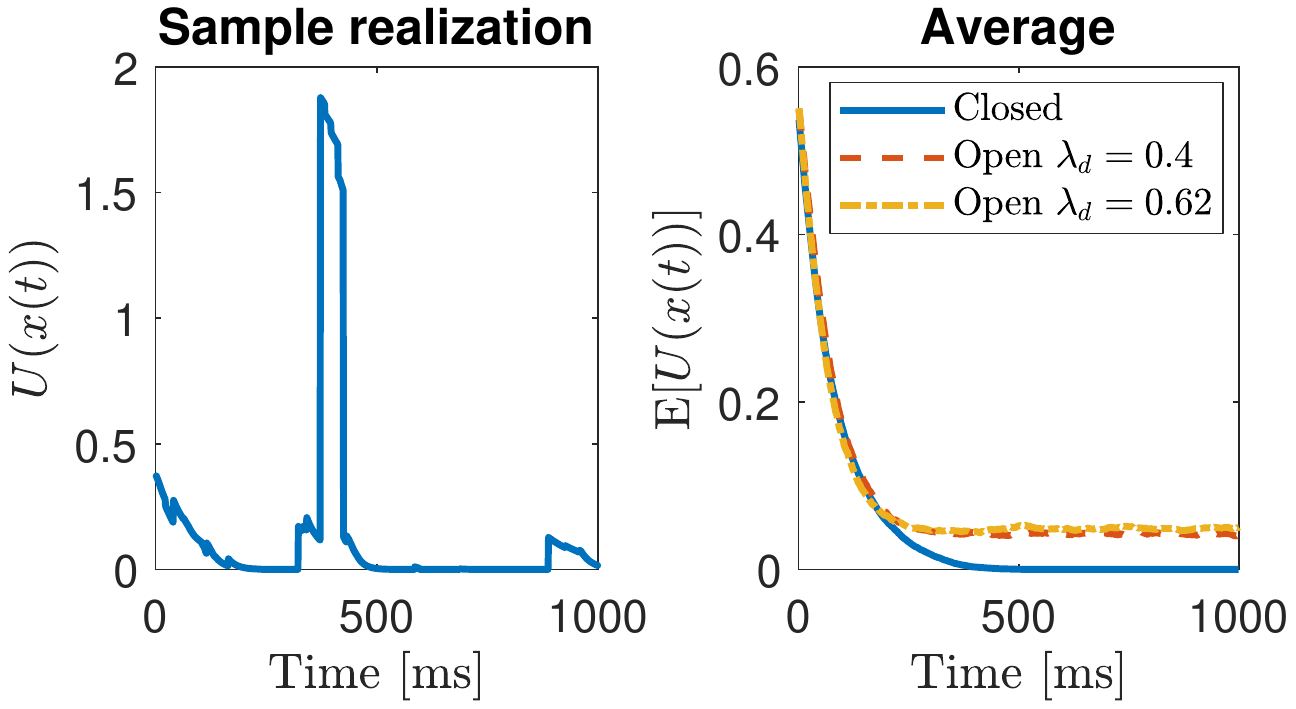}\label{fig_V_x_expectation}
}
\subfigure[$\Ep{V(x(t))}$]
{ \includegraphics[width=0.8\linewidth]{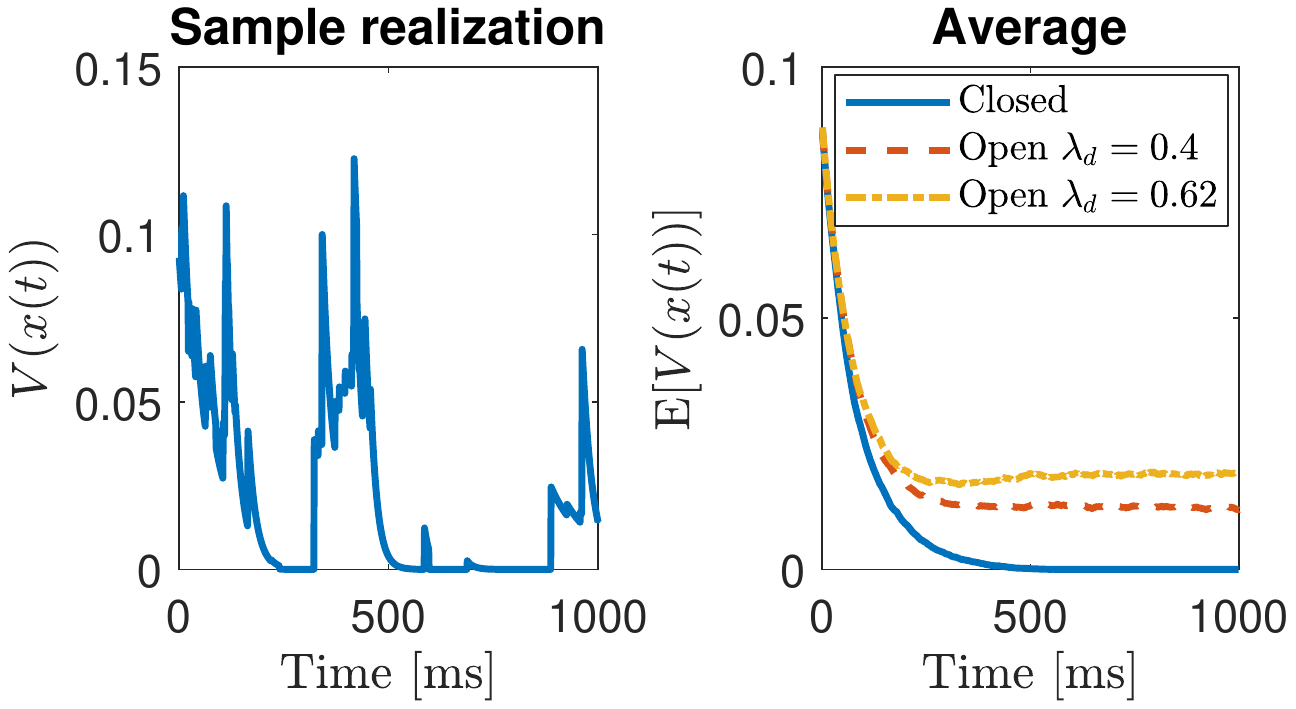}\label{fig_U_x_expectation}
}
\caption{Evolution of the Lyapunov functions for closed and open HK dynamics. On the left, sample realizations of the open dynamics. On the right, averages of 10000 simulations. }\label{fig:Lyapunov_functions}
\end{figure}

Fig.~\ref{fig:x_t_poisson} presents the trajectories of the states for one realization in the first scenario with $\lambda_d=0.4$, where we can observe the instants of arrivals and departures of agents corresponding to the appearance of new trajectories and disappearance of current trajectories respectively. For this particular realization, 19 new agents join the system and 20 departures take place, such that the number of agents at the end of the time interval considered for the simulation is 9. The initial number of clusters is 2 but new clusters appear and merge during the evolution of the dynamics, generating 3 clusters at the end. It is clear that arrivals and departures modify the number of clusters in time depending on the interval $[a,b]$ considered for the assignment of the initial states of the interacting agents.

In Fig.~\ref{fig:Lyapunov_functions}, we present the simulations of $\Ep{U_0(x(t))}$, $\Ep{U(x(t))}$ and \linebreak $\Ep{V(x(t))}$ in a closed and open scenario. In the left plots, we present the functions corresponding to the realization of the process in Fig.~\ref{fig:x_t_poisson}, that is with $\lambda_d=0.4$. In the case of $U_0(x(t))$, the discontinuities of the function are due only to arrivals and departures, while in the case of $U(x(t))$ and $V(x(t))$, the discontinuities are due to both the openness and the changes in the clusters due to the HK dynamics. 
In the right plots, the solid blue line corresponds to the simulation of the expectation of the Lyapunov functions in the closed system, the dashed red line corresponds to the simulation of the open system with $\lambda_d=0.4$ and the dash-dotted yellow line corresponds to the simulation with $\lambda_d=0.62$. 

Several observations can be made from these simulations.
Most importantly, the Lyapunov functions based on local disagreement, $U(x(t))$ and $V(x(t))$, have non-zero asymptotic values due to the arrivals and departures in the open scenarios, in contrast to the closed scenario.  Hence, the derivation of upper bounds on their asymptotic values can be useful to evaluate the impact of arrivals and departures. 

In the case of $U_0(x(t))$, instead, such an upper bound would not provide useful information, because a positive value of the function may correspond to a closed or an open system. 
Perhaps surprisingly, the asymptotic value of the Lyapunov function based on global disagreement $U_0(x(t))$, can be lower for open systems. This behavior is due to the fact that the arrival of agents may help to join different clusters of agents, decreasing the value of the variance $U_0(x(t))$ in the system. This fact is remarkable since usually the metrics associated with a system (e.g., order parameters, Lyapunov functions, mean squared errors) that are used to evaluate its performance, exhibit a worse behavior for a time-varying set of agents (see for instance \cite{vizuete2022gradient}).

\section{Conclusions and future work}\label{sect:outro}

In this paper, we introduced the problem of the HK dynamics in open multi-agent systems where agents may join and leave the network during the interactions. We formulated the problem using a stochastic setting where the arrivals and departures of agents are determined by Poisson processes. We considered several Lyapunov functions, based either on global or local disagreement, as potential tools for the analysis of this system. Finally, we provided numerical simulations for two different scenarios to show the relevance of the Lyapunov functions based on local disagreement to measure the impact of arrivals and departures in the HK dynamics.
The simulations corroborate the fact that functions based on local disagreement are promising tools to study bounded-confidence opinion dynamics in open systems.

The natural continuation of this work is the full analysis of the open HK system and, in particular, studying the asymptotic value of the Lyapunov functions through the derivation of appropriate upper bounds, depending on the arrival and departure processes. A further extension would be to analyze the performance of a social HK model where interactions are also restricted by a graph topology \cite{parasnis2022social}, extended to the case of open systems.

\bibliographystyle{IEEEtran}
\bibliography{arxiv}

\end{document}